\documentclass{ieeeaccess}
\usepackage{cite}
\usepackage{amsmath,amssymb,amsfonts}
\usepackage{graphicx}
\usepackage{bm}
\usepackage{footmisc}
\usepackage{paralist}

\newtheorem{lemma}{Lemma}

\newtheorem{theorem}{Theorem}

\newtheorem{proof}{Proof}

\def\BibTeX{{\rm B\kern-.05em{\sc i\kern-.025em b}\kern-.08em
    T\kern-.1667em\lower.7ex\hbox{E}\kern-.125emX}}

\begin{document}
\doi{}

\title{Comments on ``Time-Varying Lyapunov Functions for Tracking Control of Mechanical Systems With and Without Frictions"}
\author{\uppercase{Olalekan Ogunmolu}\authorrefmark{1}, \IEEEmembership{Member, IEEE}}
\address[1]{Microsoft Research, 300 Lafayette Street, New York, NY 10012, USA.}

\markboth
{Lekan Molu: IEEE Access}
{Author \headeretal: Preparation of Papers for IEEE TRANSACTIONS and JOURNALS}

\corresp{Corresponding author: Lekan Molu (e-mail: lekanmolu@microsoft.com).}

\begin{abstract}
	In the article\footnote{Ren, W., Zhang, B, Li, H, and Yan L. \emph{Time-Varying Lyapunov Functions for Tracking Control of Mechanical Systems With and Without Frictions}. IEEE Access, vol. 8. pp. 51510-51517. 2020.\label{ft:orig_paper}}, the authors introduced a time-varying Lyapunov function for the stability analysis of nonlinear systems whose motion is governed by standard Newton-Euler equations. The authors established asymptotic stability with the choice of two symmetric positive definite matrices  restricted by certain eigenvalue bounds in the control law. Exponential stability in the sense of Lyapunov using integrator backstepping and Lyapunov redesign is established in this note using just one matrix in the derived controller. We do not impose minimum eigenvalue bound requirements on the symmetric positive definite matrix introduced in our analysis to guarantee stability. Reducing the parameters needed in the control law, our analysis improves the  stability and convergence rates of tracking errors reported in the article\footref{ft:orig_paper}. 
\end{abstract}


\titlepgskip=-15pt

\maketitle
The well-established method of Lyapunov analysis is a great tool for synthesizing stable controllers for nonlinear systems such as robot manipulators, space flexures, parallel robots, and magnetic bearings \textit{inter alia}. The kinetics of such systems are typically derived from (time-invariant) Euler-Lagrange equations. When the underlying nonlinear system is time-varying, however, stability is difficult to establish with traditional time-invariant Lyapunov functions. 

Considering frictionless mechanical systems, the authors of the article~\cite{Ren}
\begin{inparaenum}[(i)]
	\item proposed a time-varying Lyapunov function for nonlinear systems;
	\item established sufficient conditions for their practical stability; and 
	\item under a zero friction  derived an \textit{asymptotic convergence rate for the filtered velocity and position tracking errors}.
\end{inparaenum} 
In this note, we will \textit{prove exponential error convergence when friction is absent from the dynamics}. In addition, only one feedback matrix (here the symmetric positive definite matrix $P$) is needed for the establishment of our results -- as opposed to two symmetric positive definite matrices ($P$ and $K$) reported by the authors of~\cite{Ren}.

The non-negative function used to prove asymptotic stability in the paper~\cite{Ren} is given as 
\begin{align}
	Q(t) = \dfrac{1}{2} q_r(t)^T M(q) q_r(t) + \dfrac{1}{2} \tilde{q}^T(t) P \tilde{q}(t)
\end{align} 
where all variables are as defined in the paper~\cite{Ren} i.e.  $q$ is the joint space variable, $q_r$ is the filtered velocity tracking error, defined as
\begin{align}
	q_r(t) = \dot{\tilde{q}}(t) + {\lambda} \tilde{q}(t),
	\label{eq:filtered_error}
\end{align}
for a  ${\lambda} > 0$; $\tilde{q}(t)$ is the position tracking error given as 
%
	$\tilde{q}(t) = q_d(t)- q(t)$,
%
for a desired joint position, $q_d$. The matrix $M(q)$ is the symmetric positive definite inertia matrix of the system. 

The model of the system to be controlled is given by 
\begin{equation}
	M(q) \ddot{q} + C(q, \dot{q})\dot{q} + G(q)+ F(t, \dot{q})=\tau
	\label{eq:system}
\end{equation}
where $C(q, \dot{q})$ denotes the Coriolis and centrifugal forces, $G(q)$ models the gravity and non-conservative external forces, $F(t, \dot{q})$ is the time-varying friction, and $\tau$ is the vector of actuator torques. 

\textbf{Observe}: The matrix $(\dot{M}(q) - 2C(q, \dot{q}))$ is skew-symmetric for a suitable choice of $C(q, \dot{q})$. One may also write this property as $\dot{M}(q) = C(q, \dot{q}) + C^T(q, \dot{q})$.
\begin{lemma}[Asymptotic stability without frictions]
	For the nonlinear manipulator \eqref{eq:system} with a zero time-varying friction i.e. $F(t, \dot{q})=0$, every motion of the mechanical system is asymptotically stable when $F(t, \dot{q})\equiv 0$, i.e. $q(t) \rightarrow q_d(t)$ as $t \rightarrow \infty$ provided we choose a virtual input, $\eta(t)$, which is backstepped as an integrator as follows
	\begin{subequations}
		\begin{align}
			\eta (t) &= \dot{q}(t), \\ 
			&\triangleq \dot{q}_d(t)- \tilde{q}(t) - \dfrac{1}{2} M^{-1}(q)\dot{M}(q)\tilde{q}(t).
			\label{eq:virtual_cancel}
		\end{align}
	\end{subequations}  
\end{lemma}
\begin{proof}	
	Define the tracking error dynamics for a desired joint space variable, $q_d$, as 
	\begin{align}
		\dot{\tilde{q}}(t) &= \eta (t) - \dot{q}_d(t).
		\label{eq:error_1_def}
	\end{align}
	Now, consider the following Lyapunov function candidate 
	\begin{align}
		V_1(t) = \frac{1}{2} \tilde{q}^T(t) {M(q)} \tilde{q}(t).
		\label{eq:lyap_1}
	\end{align} 
	We find that 
	\begin{align}
		\dot{V}_1(t) &= \tilde{q}^T(t) {M}(q) \dot{\tilde{q}}(t) + \dfrac{1}{2} \tilde{q}^T(t) \dot{M}(q) \tilde{q}(t), \\
		&= \tilde{q}^T(t) {M}(q) (\eta (t) - \dot{q}_d(t))+ \dfrac{1}{2} \tilde{q}^T(t) \dot{M}(q) \tilde{q}(t),
		\\
		&= -\tilde{q}^T(t) {M}(q) \tilde{q}(t).
		\label{eq:subsys1_stable}
	\end{align} 
	Equation \eqref{eq:subsys1_stable} implies that the system must be \textit{asymptotically stable}~\cite{Petros}. 
\end{proof}

\begin{theorem}[Exponential stability of the system]
	Consider the entire nonlinear system \eqref{eq:system}, suppose that there exists a $P =P^T \, \succ 0$, then trajectories that emanate from the system are \textbf{exponentially stable} if we choose the control law, $\tau$, according to 
	\begin{align}
		\tau &= {M}(q)\left(\ddot{q_d} + \lambda\dot{\tilde{q}}+ q_r\right) + C(q,\dot{q})\left(\dot{q}+q_r\right) 
		+ G(q).
		\label{eq:control_nom}
	\end{align}
	%
	\label{th:exp_stab}
\end{theorem}

\begin{proof}
	Consider the Lyapunov function candidate
	\begin{align}
		V(t) = V_1(t) + \frac{1}{2}q_r^T(t) PM(q) q_r(t)
	\end{align}
	so 	that the time derivative of $V(t)$ (dropping the r.h.s templated arguments for notation terseness) is
	\begin{subequations}
		\begin{align}
			&\dot{V}(t) = \dot{V}_1(t) + q_r^T PM\dot{q}_r + \frac{1}{2}q_r^T P\dot{M} q_r, \\
			&= \dot{V}_1(t) + q_r^T PM (\ddot{\tilde{q}}+\lambda\dot{\tilde{q}}
			) + \frac{1}{2}q_r^T P\dot{M} q_r, \\
			&= \dot{V}_1(t) + q_r^T PM (\ddot{q}_d-\ddot{q}+\lambda\dot{\tilde{q}}) + \frac{1}{2}q_r^T P\dot{M} q_r, \\
			&= \dot{V}_1(t) + q_r^T PM (\ddot{q}_d+\lambda \dot{\tilde{q}})
			\nonumber \\&\qquad\qquad
			 +   q_r^T P\left[C(q, \dot{q}) (\dot{q}+ {q}_r) +G(q)-\tau\right].
			\label{eq:lyap_deriv_exp}
		\end{align}
	\end{subequations}
	Where we have used the skew-symmetric property $\dot{M} - 2 C$ in arriving at \eqref{eq:lyap_deriv_exp}. Let $\tau$ be as defined in \eqref{eq:control_nom}, 
	%
	then we must have, 
	\begin{subequations}
		\begin{align}
			\dot{V}(t) &= -\tilde{q}(t)M(q)\tilde{q}(t) - q_r^T(t)PM(q) q_r(t), \\
			\dot{V}(t) & \triangleq -2 V(t).
			\label{eq:lyap_exp_proof}
		\end{align}
	\end{subequations}
	\textit{A fortiori}, exponential stability follows~\cite{Petros} as a result of \eqref{eq:lyap_exp_proof}.
\end{proof}

Through a careful choice of virtual input design, the convergence rate of system tracking errors can be significantly improved upon. In addition, the number of parameters introduced into the torque control equation is reduced from two to one. This is particularly of importance when designing fast real-time learning-based optimization for suitable controllers for nonlinear systems for which equation \eqref{eq:system} is applicable.


\EOD

\end{document}